\documentclass[journal,draftcls, onecolumn]{IEEEtran}
\usepackage[dvipdfmx]{graphicx}
\usepackage{amsmath}
\usepackage{amssymb}
\usepackage{amsfonts}
\usepackage{mathrsfs}
\usepackage{theorem}
\usepackage{color}
\usepackage{float}
\usepackage{subfig}
\DeclareMathAlphabet{\bm}{OML}{cmm}{b}{it}
\theorembodyfont{\rmfamily}
\newtheorem{theorem}{Theorem}
\newtheorem{lemma}{Lemma}

\newtheorem{remark}{Remark}
\newtheorem{proposition}{Proposition}

\newcommand{\qed}{\hfill \IEEEQED}

\newcommand{\rom}[1]{\mathrm{#1}}

\newcommand{\Pc}{\rom{P}_{\rom{c}}}

\renewcommand{\theenumi}{(\roman{enumi})}

\begin{document}

\title{On The Most Discriminative Boolean Functions for Correlated Sources}

\author{
\IEEEauthorblockN{Jun Chen, Shun Watanabe, and Lei Yu}
\thanks{Authors are listed alphabetically.}
}


\maketitle
\begin{abstract}
Motivated by a conjecture of Amari and Kobayashi, we study the problem of
identifying pairs of Boolean functions that maximize the Kullback-Leibler divergence
between two distributions obtained by separately compressing two correlated sources.
When the reference distribution corresponds to independent sources, this problem reduces
to the problem of maximizing mutual information, for which the optimality of dictator functions 
has been proved by Pichler, Piantanida, and Matz. For the problem of maximizing Fisher information,
which can be viewed as a local version of the problem studied in this paper, Amari and Kobayashi
conjectured that parity functions are optimal. 
For unbiased pairs of Boolean functions, and for identical pairs in the nonnegative correlation regime, 
we prove that both the divergence and the Fisher information are
maximized by level-$k$ functions, namely,
functions whose Fourier coefficients are supported only on level $k$. 
Since level-$k$ functions include parity functions, this gives a partial resolution of the conjecture of
Amari and Kobayashi. 
Furthermore, in the framework of Bayesian distributed one-bit hypothesis testing,
we prove that level-$k$ functions are optimal among all pairs of functions. 
Finally, we also discuss the one function version of the problem studied in this paper,
which can be regarded as the divergence analogue of the Courtade and Kumar conjecture. 
\end{abstract}

\section{Introduction}

Toward resolving open problems in distributed statistical inference \cite{han-amari:98},
Amari considered the following simple problem \cite{Amari:11}. Suppose that two  correlated binary
sources $(X^n,Y^n)$ follow a $\rho$-correlated distribution; see Fig.~\ref{Fig:sysmte}. If the sources $X^n$ and
$Y^n$ are compressed by Boolean functions $f$ and $g$, respectively, which pairs of functions
maximize the Fisher information of the resulting output distribution?

Amari proved that, when $\rho=0$, dictator functions are optimal;\footnote{More precisely,
Amari showed the optimality of dictator function for one function version of the Fisher information maximization problem.} when $\rho$ is away from zero,
he demonstrated that majority functions can outperform dictator functions \cite{Amari:11}.
Furthermore, Amari and Kobayashi observed that parity functions of multiple bits can perform
well when $\rho$ is away from zero, and posed the conjecture that the Fisher information 
is maximized by parity functions for every $\rho$ \cite{Amari:23,Amari-biograph,KobAma:23}. 
Indeed, as $\rho$ moves away from zero, taking 
the parity of $k\ge 2$ bits can yield a larger Fisher information than a dictator function, and the value of $k$ that
maximizes the Fisher information depends on $\rho$. Since parity functions weaken the correlation 
between $X^n$ and $Y^n$ compared to dictator functions, 
it is not intuitively clear why parity functions outperform dictator functions, 
making their conjecture particularly interesting.

Amari's problem is rooted in one bit compression
in distributed statistical inference and distributed hypothesis testing \cite{han-amari:98}.
On the other hand, the Fourier analysis of Boolean functions has been actively studied 
in theoretical computer science \cite{ODonnell-book}. Furthermore, over the past decade, Fourier analysis has been effectively used in information theory to study the problem of identifying Boolean functions that maximize quantities such as mutual information, often referred to as the most informative Boolean function problem \cite{AnaGohKamNai13b,CouKum:14,OrdShaWei:16,Sam:16,AnaBogChaJayNai:17,PicPiaMat:18,Pichler:thesis,BarOzg:20,LiMed:21,Yu-Tan-book,Yu:23,CheNai:24,CheGohNai:25,JavWoo:26}. 
These lines of work motivate us to study
the conjecture of Amari and Kobayashi 
through the lens of Fourier analysis of Boolean functions.

In this paper, we consider the problem of identifying the most discriminative Boolean functions in the following sense.
Suppose that two correlated binary sources $(X^n,Y^n)$ follow either $\rho_0$-correlated distribution or $\rho_1$-correlated distribution; see Fig.~\ref{Fig:sysmte}.
When the sources $X^n$ and $Y^n$ are compressed by Boolean functions $f$ and $g$, respectively, 
we ask which pairs of functions maximize the Kullback-Leibler (KL) divergence of the resulting output distributions?
When the reference distribution corresponds to independent sources, i.e., $\rho_1=0$, this problem reduces to the 
problem of maximizing mutual information, for which the optimality of dictator function has been proved by
Pichler, Piantanida, and Matz in \cite{PicPiaMat:18} (see also \cite{Pichler:thesis}). Also, the aforementioned problem of Amari can be viewed 
as a local version, i.e., in the limit where $\rho_0$ and $\rho_1$ approach each other, of the problem studied in this paper. 
In fact, for certain values of the parameters $(\rho_0,\rho_1)$, especially when $\rho_1$ is away from zero, 
we observe that parity functions yield a larger divergence than dictator functions.

In this paper, we prove that, when the functions $f$ and $g$ are unbiased, or when $f$ and $g$ are identical (not necessarily unbiased),
both the divergence and the Fisher information are maximized by level-$k$ functions, namely functions 
whose Fourier coefficients are supported only on level $k$. Since level-$k$ functions include parity functions,
this gives a partial resolution of the conjecture by Amari and Kobayashi. However, while every level-$1$ function is
necessarily a dictator function, level-$k$ functions for $k\ge 2$ are not necessarily parity functions.
Thus, the optimal value is attained by a class of functions that is broader than the class of parity functions.

As a more operationally defined problem,
we also study a problem of Bayesian one bit distributed hypothesis testing. 
In this problem, we prove that the Bayes error probability is minimized by level-$k$ functions
among all pairs of Boolean functions. 

As mentioned above, when $\rho_1=0$, the divergence reduces to the mutual information.
Thus, as a natural extension of the so-called Courtade-Kumar conjecture \cite{CouKum:14}, which concerns the one function
version of the mutual information maximization problem, we can consider the one function version of the divergence maximization problem.
However, for certain values of parameters $(\rho_0,\rho_1)$, even within the class of unbiased functions, the maximum can
be attained by functions other than level-$k$ functions. This suggest that the divergence maximization problem
exhibits different behavior in the two function and one function settings. 

In a related line of work, Girish, Cung, and Telatar studied binary distributed hypothesis
testing under linear compression schemes and showed that simple truncation is
the best linear scheme for testing correlations of opposite signs with the
same magnitude and for testing for or against independence; they further
conjectured, based on numerical evidence, that truncation remains optimal among
linear codes whenever the two correlations have opposite signs \cite{GirCunTel:26}.
In the one-bit
case, truncation corresponds to a dictator function.

The rest of the paper is organized as follows.
In Section \ref{section:Fourier}, we review some results on the Fourier analysis on Boolean cube that are used in this paper. 
In Section \ref{section:KL} and Section \ref{section:Fisher}, we consider the divergence maximization problem and the Fisher information maximization
problem, respectively. 
In Section \ref{section:Bayes}, we consider the Bayesian distributed hypothesis testing by Boolean functions.
We conclude the paper with some discussion on the one function problem in Section \ref{section:discussion}.

\begin{figure}[t]
\centering{
\includegraphics[width=0.5\textwidth]{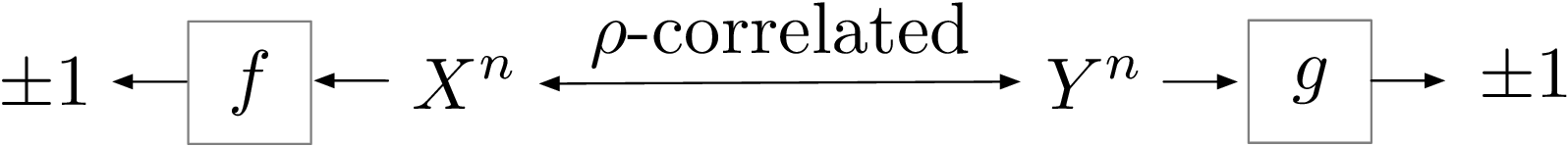}
\caption{A description of the problems studied by Amari and this paper. In Amari's problem, $(X^n,Y^)$
follow $\rho$-correlated distribution, and we are interested in identifying functions $f$ and $g$ that maximize the Fisher
information of the resulting output distribution $P_{f(X^n)g(Y^n),\rho}$. On the other hand, in the problem studied in this paper,
$(X^n,Y^n)$ follow either $\rho_0$ correlated distribution or $\rho_1$-correlated distribution, and we are interested in identifying functions 
$f$ and $g$ that maximize the KL divergence between the resulting output distributions $P_{f(X^n)g(Y^n),\rho_0}$ and 
$P_{f(X^n)g(Y^n),\rho_1}$.}
\label{Fig:sysmte}
}
\end{figure}

\section{Fourier Analysis} \label{section:Fourier}

In this section, we review some terminology and basic results in the Fourier analysis on Boolean cube; for further details, see \cite{ODonnell-book}.
Let $\mathbb{F}_2=\{0,1\}$ denote the finite field with two elements. For two elements $x,y \in \mathbb{F}_2$, we write 
$x+y$ and $xy$ for their addition and multiplication, respectively. For vectors $\bm{x}=(x_1,\ldots, x_n)$ and $\bm{y}=(y_1,\ldots,y_n)$
in $\mathbb{F}_2^n$, we denote the component-wise addition by $\bm{x}+\bm{y}$, and the inner product by $\bm{x}\cdot \bm{y}=\sum_{i=1}^n x_iy_i$.
For a subset $S \subset [n]:=\{1,\ldots,n\}$, we define parity function $\chi_S: \mathbb{F}_2^n \to \mathbb{R}$ by
\begin{align*}
\chi_S(\bm{x}) := (-1)^{\sum_{i\in S}x_i} = \prod_{i\in S} (-1)^{x_i}.
\end{align*}
For two functions $f,g:\mathbb{F}_2^n \to \mathbb{R}$, we define the inner product $\langle \cdot,\cdot\rangle$ by
\begin{align*}
\langle f,g\rangle := \sum_{\bm{x} \in \mathbb{F}_2^n} \frac{1}{2^n} f(\bm{x})g(\bm{x}).
\end{align*}
For two subset $S,T \subset [n]$, we can verify that
\begin{align} \label{eq:parity-orthogonality}
\langle \chi_S,\chi_T\rangle = \left\{
\begin{array}{ll}
1 & \mbox{if } S = T \\
0 & \mbox{if } S \neq T
\end{array}
\right..
\end{align}
Thus, the set of parity functions $\chi_S$ for $S \in 2^{[n]}$ form an orthonormal basis 
of the vector space of all functions from $\mathbb{F}_2^n$ to $\mathbb{R}$.
This basis is called Fourier basis on the Boolean cube $\mathbb{F}_2^n$, and we can expand a function $f:\mathbb{F}_2^n\to \mathbb{R}$ as
\begin{align*}
f(\bm{x}) = \sum_{S \subset [n]} \widehat{f}(S) \chi_S(\bm{x}),
\end{align*}
where 
\begin{align*}
\widehat{f}(S) = \langle f, \chi_S \rangle
\end{align*}
is a Fourier coefficient. Note that 
\begin{align} \label{eq:fourier-coefficient-empty}
\widehat{f}(\emptyset) = \mathbb{E}[f(X^n)]
\end{align}
for the uniform random vector $X^n$ on $\mathbb{F}_2^n$.

Since the parity functions form an orthonormal basis, the following Plancherel identity holds:
\begin{align} \label{eq:Plancherel}
\langle f, g \rangle = \sum_{S \subset [n]} \widehat{f}(S) \widehat{g}(S).
\end{align}
Particularly, when $f=g$, it is called the Parseval identity:
\begin{align} \label{eq:Parseval}
\langle f,f \rangle = \sum_{S \subset [n]} \widehat{f}(S)^2.
\end{align}
When a function is Boolean $f: \mathbb{F}_2^n \to \{\pm 1\}$, then 
$\langle f, f \rangle=1$ and the Parseval identity \eqref{eq:Parseval} imply
\begin{align} \label{eq:Parseval-Boolean}
\sum_{S \subset [n]} \widehat{f}(S)^2 = 1.
\end{align}
Thus, the squared Fourier coefficients can be regarded as a probability distribution.

Let $X^n = (X_1,\ldots,X_n)$ be a uniformly distributed random vector on $\mathbb{F}_2^n$.
For parameter $\rho \in [-1,1]$, let $Y^n$ be such that, for each $i \in [n]$ independently, 
\begin{align} \label{eq:rho-correlated}
Y_i = \left\{
\begin{array}{ll}
X_i & \mbox{with probability } \frac{1+\rho}{2} \\
X_i + 1 & \mbox{with probability } \frac{1-\rho}{2}
\end{array}
\right..
\end{align}
We say that $Y^n$ is $\rho$-correlated with $X^n$, or $(X^n,Y^n)$ is a $\rho$-correlated pair. 

For $\rho \in [-1,1]$, the noise operator with parameter $\rho$ is the linear operator $T_\rho$ on functions $f:\mathbb{F}_2^n \to \mathbb{R}$ 
defined by the conditional expectation
\begin{align*}
T_\rho f(\bm{x}) := \mathbb{E}[f(Y^n) | X^n = \bm{x}],
\end{align*}
where $(X^n,Y^n)$ is a $\rho$-correlated pair. For a parity function $\chi_S$, we have
\begin{align}
T_\rho \chi_S(\bm{x}) = \mathbb{E}[\chi_S(Y^n) | X^n = \bm{x}] = \prod_{i \in S} \mathbb{E}[ (-1)^{Y_i} | X_i = x_i] = \prod_{i \in S} \rho (-1)^{x_i} = \rho^{|S|} \chi_S(\bm{x}).
\label{eq:affect-of-noise-operator}
\end{align}
Thus, by linearity, we have
\begin{align} \label{eq:noise-operator-effect}
T_\rho f = \sum_{S \subset [n]} \rho^{|S|} \widehat{f}(S) \chi_S.
\end{align}

For $\rho \in [-1,1]$, the noise stability of function $f:\mathbb{F}_2^n \to \mathbb{R}$ at $\rho$ is defined by
\begin{align*}
\mathbf{Stab}_\rho[f] := \mathbb{E}[ f(X^n) f(Y^n)], 
\end{align*}
where $(X^n,Y^n)$ is a $\rho$-correlated pair. 
If $f:\mathbb{F}_2^n \to \{\pm 1\}$ is a Boolean function, we have
\begin{align*}
\mathbf{Stab}_\rho[f] &= \Pr\big( f(X^n) = f(Y^n) \big) - \Pr\big( f(X^n) \neq f(Y^n) \big) \\
&= 2 \Pr\big( f(X^n) = f(Y^n) \big) -1.
\end{align*}
By \eqref{eq:noise-operator-effect} and the orthonormality of parity functions \eqref{eq:parity-orthogonality}, we have
\begin{align}
\mathbf{Stab}_\rho[f] = \langle f, T_\rho f \rangle = \sum_{S \subset [n]} \rho^{|S|} \widehat{f}(S)^2.
\end{align}
Similarly, for two functions $f,g: \mathbb{F}_2^n \to \mathbb{R}$ and a $\rho$-correlated pair $(X^n,Y^n)$, we have
\begin{align}
\mathbb{E}[f(X^n) g(Y^n)] = \langle f, T_\rho g \rangle = \sum_{S \subset [n]} \rho^{|S|} \widehat{f}(S) \widehat{g}(S).
\end{align}
If $f$ and $g$ are Boolean, then we have
\begin{align}
\mathbb{E}[f(X^n) g(Y^n)] &=  \Pr\big( f(X^n) = g(Y^n) \big) - \Pr\big( f(X^n) \neq g(Y^n) \big) \\
&= 2 \Pr\big( f(X^n) = g(Y^n) \big) -1.
\end{align}

For Boolean functions $f$ and $g$,
\eqref{eq:fourier-coefficient-empty} implies
\begin{align}
a_{\pm 1} &:= \Pr\big( f(X^n) = \pm 1 \big) = \frac{1\pm \widehat{f}(\emptyset)}{2}, \label{eq:marginal-1} \\
b_{\pm 1} &:= \Pr\big( g(Y^n) = \pm 1 \big) = \frac{1 \pm \widehat{g}(\emptyset)}{2}. \label{eq:marginal-2}
\end{align}
Let
\begin{align}
\theta_\rho &:= \frac{1}{4}\big( \langle f, T_\rho g \rangle - \widehat{f}(\emptyset) \widehat{g}(\emptyset) \big) \label{eq:theta-definition} \\
&= \frac{1}{4} \sum_{S \subset [n]: \atop S \neq \emptyset} \widehat{f}(S) \widehat{g}(S) \rho^{|S|} . \label{eq:theta-definition-2} 
\end{align}
Then, by solving \eqref{eq:marginal-1}, \eqref{eq:marginal-2}, and 
\begin{align*}
\Pr\big( f(X^n) = g(Y^n) \big) = \frac{1+\langle f, T_\rho g \rangle}{2},
\end{align*}
the joint distribution of $f(X^n)$ and $g(Y^n)$ is given by
\begin{align} \label{eq:joint-distribution-fg}
\Pr\big( f(X^n) = u,~g(Y^n)=v \big) = a_u \cdot b_v + u\cdot v \cdot \theta_\rho, 
\end{align}
where $u,v \in \{\pm 1\}$. Particularly, when $f$ and $g$ are unbiased, i.e., $\widehat{f}(\emptyset)=\widehat{g}(\emptyset)=0$, then
\begin{align} \label{eq:joint-distribution-fg-unbiased}
\Pr\big( f(X^n) = u,~g(Y^n)=v \big) = \frac{1 + u\cdot v \cdot \sum_{S\subset [n]:\atop S\neq \emptyset }  \widehat{f}(S)\widehat{g}(S) \rho^{|S|}}{4}.
\end{align}

When Fourier coefficients of a Boolean function $f$ are concentrated on sets of size $k \in [n]$, i.e., $\widehat{f}(S) = 0$ unless $|S|=k$,
then we shall refer to such a function as a level-$k$ function. 
When $f=g$ is a level-$k$ function for $k\ge 1$,\footnote{Level-$0$ function is a constant function.} \eqref{eq:joint-distribution-fg} implies
\begin{align} \label{eq:joint-of-level-k-function}
\Pr\big( f(X^n) = u,~f(Y^n) = v \big) = \frac{1+u\cdot v \cdot \rho^k}{4}.
\end{align}

It is known that a level-$1$ function must be 
a dictator function $\pm \chi_{\{i\}}$ 
for some $i \in [n]$ (eg.~see \cite[Exercise 1.19]{ODonnell-book}). However, a level-$k$ function for $k \ge 2$ needs not be 
a parity function $\pm \chi_S$ for $|S|=k$. For instance, the function given by
\begin{align*}
f(x_1,x_2,x_3,x_4) &= \frac{1}{2} \cdot (-1)^{x_1+x_3} + \frac{1}{2} \cdot (-1)^{x_1+x_4} + \frac{1}{2}\cdot (-1)^{x_2+x_3} - \frac{1}{2} \cdot (-1)^{x_2+x_4} \\
&= (-1)^{x_1} \cdot \frac{(-1)^{x_3} + (-1)^{x_4}}{2} + (-1)^{x_2} \cdot \frac{(-1)^{x_3} - (-1)^{x_4}}{2}
\end{align*}
is a level-$2$ function that is not a parity function.


\section{KL-Divergence Setting} \label{section:KL}


For $\rho \in (-1,1)$, let $P_{X^nY^n,\rho}$ be the joint distribution
of a $\rho$-correlated pair $(X^n,Y^n)$ given by \eqref{eq:rho-correlated}.
Then, for two Boolean functions $f,g:\mathbb{F}_2^n \to \{+1,-1\}$, let $P_{f(X^n)g(Y^n),\rho}$
be the joint distribution of $(f(X^n),g(Y^n))$.
In this section, for given parameters $\rho_0,\rho_1 \in (-1,1)$, we are interested in
Boolean functions $f$ and $g$ that maximizes the KL-divergence between
$P_{f(X^n)g(Y^n),\rho_0}$ and $P_{f(X^n)g(Y^n),\rho_1}$. In fact, when $f=g$ is a level-$k$ function for $k\ge1$, then \eqref{eq:joint-of-level-k-function}
implies 
\begin{align} \label{eq:divergence-level-k}
D(P_{f(X^n)f(Y^n),\rho_0} \| P_{f(X^n)f(Y^n),\rho_1}) = d\bigg(\frac{1+\rho_0^k}{2} \bigg\| \frac{1+\rho_1^k}{2}\bigg),
\end{align}
where 
\begin{align*}
D(P\|Q) = \sum_z P(z)\ln\frac{P(z)}{Q(z)}
\end{align*}
is the KL divergence between two distributions $P$ and $Q$ on the same alphabet, and where 
$d(p\|q) = p\ln \frac{p}{q}+(1-p)\ln \frac{(1-p)}{(1-q)}$ is the binary divergence function. 
The optimal value of $k$ that maximizes \eqref{eq:divergence-level-k} depends on parameters $\rho_0$ and $\rho_1$; see Fig.~\ref{Fig:Divergence1}.
Our goal is to prove whether the divergence is optimized by a level-$k$ function for some $k$. 

\begin{figure}[t]
\centering{
\includegraphics[width=0.5\textwidth]{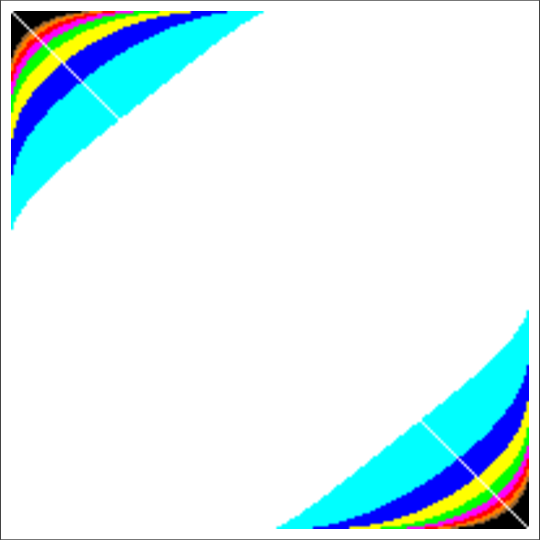}
\caption{A plot of $1 \le k \le 10$ that maximizes \eqref{eq:divergence-level-k} for $-1 \le \rho_0, \rho_1 \le 1$; 
the horizontal axis is $\rho_1$ and the vertical axis is $\rho_0$;
the left-top is $(-1,-1)$ and the right-bottom is $(1,1)$; White $\to 1$, Cyan $\to 2$, Blue $\to 3$, Yellow $\to 4$, Green $\to 5$, Magenta $\to 6$, Red $\to 7$,
Orange $\to 8$, Brown $\to 9$, and Black $\to 10$.}
\label{Fig:Divergence1}
}
\end{figure}

Among the class of unbiased Boolean functions, we can show that the divergence is maximized when
$f$ and $g$ are identical level-$k$ functions.
\begin{theorem}[Unbiased Functions] \label{theorem:unbiased}
Let $\rho_0,\rho_1 \in (-1,1)$.
Let $f,g:\mathbb{F}_2^n\to\{\pm 1\}$ be unbiased Boolean functions, i.e., $\widehat{f}(\emptyset)=\widehat{g}(\emptyset)=0$. Then, we have
\begin{align} \label{eq:theorem:unbiased}
D(P_{f(X^n)g(Y^n),\rho_0} \| P_{f(X^n)g(Y^n),\rho_1}) \le \max_{1 \le k \le n} d\bigg(\frac{1+\rho_0^k}{2} \bigg\| \frac{1+\rho_1^k}{2}\bigg).
\end{align}
\end{theorem}

When $X^n$ and $Y^n$ are non-negatively correlated and either $g=f$ or $g=-f$, we can show that
the divergence is maximized by level-$k$ functions.
\begin{theorem}[Biased Function] \label{theorem:biased}
Let $\rho_0,  \rho_1 \in [0,1)$. Let $f:\mathbb{F}_2^n\to \{\pm 1\}$ be a Boolean function (not necessarily unbiased).
Then, we have
\begin{align} \label{eq:ff}
D(P_{f(X^n)f(Y^n),\rho_0} \| P_{f(X^n)f(Y^n),\rho_1}) \le \max_{1 \le k \le n} d\bigg(\frac{1+\rho_0^k}{2} \bigg\| \frac{1+\rho_1^k}{2}\bigg).
\end{align}
Furthermore, for $g=-f$, we have
\begin{align} \label{eq:negative-f}
D(P_{f(X^n)g(Y^n),\rho_0} \| P_{f(X^n)g(Y^n),\rho_1}) \le \max_{1 \le k \le n} d\bigg(\frac{1+\rho_0^k}{2} \bigg\| \frac{1+\rho_1^k}{2}\bigg).
\end{align}
\end{theorem}

Although it has not been proved that the divergence is always maximized by level-$k$ functions,
we can show that level-$k$ functions are locally optimal in the following sense:
if one of the two functions is a level-$k$ function, then the divergence cannot exceed the value
obtained by choosing the other function to be the same level-$k$ function. 
\begin{theorem}[Local Optimality] \label{theorem:local-optimality}
Let $\rho_0,\rho_1 \in (-1,1)$.
Let $f:\mathbb{F}_2^n\to \{\pm 1\}$ be a level-$k$ function. Then, for any Boolean function $g$, we have
\begin{align*}
D(P_{f(X^n)g(Y^n),\rho_0} \| P_{f(X^n)g(Y^n),\rho_1}) \le D(P_{f(X^n)f(Y^n),\rho_0} \| P_{f(X^n)f(Y^n),\rho_1}).
\end{align*}
\end{theorem}
Note that Theorem \ref{theorem:local-optimality} holds for any $k$, regardless of whether it maximizes the
divergence in \eqref{eq:divergence-level-k}.

\subsection{Proof of Theorem \ref{theorem:unbiased}}

To prove Theorem \ref{theorem:unbiased}, we first prove the following lemma; 
it can be regarded as a variation of joint convexity of the KL divergence 
that involves two weight vectors, which may be of independent interest. 
\begin{lemma} \label{lemma:inner-product-convexity}
Let $a = (a_1,\ldots,a_m), b= (b_1,\ldots,b_m) \in \mathbb{R}^m$ be unit vectors, i.e., $\sum_{i=1}^m a_i^2 =  \sum_{i=1}^m b_i^2  = 1$.
Then, for $\mu = (\mu_1,\ldots,\mu_m), \nu = (\nu_1,\ldots,\nu_m) \in (-1,1)^m$, we have\footnote{Note that, by the Cauchy-Schwarz inequality,
$\sum_{i=1}^m a_i b_i \mu_i \in (-1,1)$ and $\sum_{i=1}^m a_i b_i \nu_i \in (-1,1)$.}
\begin{align}
d\bigg( \frac{1 + \sum_{i=1}^m a_i b_i \mu_i}{2} \bigg\| \frac{1+ \sum_{i=1}^m a_i b_i \nu_i}{2} \bigg) 
\le \min\bigg[
\sum_{i=1}^m a_i^2 d\bigg( \frac{1+\mu_i}{2} \bigg\| \frac{1+\nu_i}{2} \bigg),~\sum_{i=1}^m b_i^2 d\bigg( \frac{1+\mu_i}{2} \bigg\| \frac{1+\nu_i}{2} \bigg)
\bigg].
\end{align}
\end{lemma}
\begin{proof}
For $t \in [0,1]$, let 
\begin{align*}
\psi(t) &:= d\bigg( \frac{1+ \sum_{i=1}^m a_i b_i(\nu_i + t(\mu_i-\nu_i))}{2} \bigg\| \frac{1+\sum_{i=1}^m a_i b_i \nu_i}{2}\bigg), \\
\phi(t) &:= \sum_{i=1}^m a_i^2 d\bigg( \frac{1+(\nu_i + t(\mu_i-\nu_i))}{2} \bigg\| \frac{1+\nu_i}{2}\bigg), \\
\varphi(t) &:= \sum_{i=1}^m b_i^2 d\bigg( \frac{1+(\nu_i + t(\mu_i-\nu_i))}{2} \bigg\| \frac{1+\nu_i}{2} \bigg).
\end{align*}
Then, the claim of the lemma is equivalently stated as 
\begin{align*}
\psi(1) \le \min[ \phi(1), \varphi(1)].
\end{align*}
Clearly, we have
\begin{align*}
\psi(0) = \phi(0) = \varphi(0) = 0.
\end{align*}
Moreover, we have
\begin{align*}
\psi^\prime(t) &= \frac{\sum_{i=1}^m a_i b_i (\mu_i-\nu_i)}{2} \ln \frac{\big(1+\sum_{i=1}^m a_i b_i (\nu_i + t(\mu_i-\nu_i))\big)\big(1-\sum_{i=1}^m a_i b_i \nu_i\big)}{\big(1-\sum_{i=1}^m a_i b_i (\nu_i + t(\mu_i-\nu_i))\big) \big(1+ \sum_{i=1}^m a_i b_i \nu_i\big)}, \\
\phi^\prime(t) &= \sum_{i=1}^m \frac{a_i^2(\mu_i-\nu_i)}{2} \ln \frac{\big( 1+ (\nu_i + t(\mu_i - \nu_i))\big)\big(1-\nu_i\big)}{\big(1- (\nu_i + t(\mu_i-\nu_i))\big)\big(1+\nu_i \big)}, \\
\varphi^\prime(t) &= \sum_{i=1}^m \frac{b_i^2(\mu_i-\nu_i)}{2} \ln \frac{\big(1+(\nu_i+t(\mu_i-\nu_i))\big)\big(1-\nu_i\big)}{\big(1-(\nu_i+t(\mu_i-\nu_i))\big)\big(1+\nu_i\big)}.
\end{align*}
In particular, we have
\begin{align*}
\psi^\prime(0) = \phi^\prime(0) = \varphi^\prime(0) = 0.
\end{align*}
Therefore, it suffices to show 
\begin{align} \label{eq:proof-inner-product-convexity-0}
\psi^{\prime\prime}(t) \le \min[ \phi^{\prime\prime}(t), \varphi^{\prime\prime}(t)]
\end{align}
for $t \in [0,1]$. Note that
\begin{align*}
\psi^{\prime\prime}(t) &= \frac{\big(\sum_{i=1}^m a_i b_i (\mu_i-\nu_i)\big)^2}{1 - \big(\sum_{i=1}^m a_i b_i (\nu_i + t(\mu_i - \nu_i))\big)^2}, \\
\phi^{\prime\prime}(t) &= \sum_{i=1}^m \frac{a_i^2 (\mu_i-\nu_i)^2}{1 - (\nu_i +t(\mu_i-\nu_i))^2}, \\
\varphi^{\prime\prime}(t) &= \sum_{i=1}^m \frac{b_i^2(\mu_i-\nu_i)^2}{1 - (\nu_i + t(\mu_i-\nu_i))^2}.
\end{align*}

Let $x_i := \frac{a_i(\mu_i-\nu_i)}{\sqrt{1-(\nu_i+t(\mu_i-\nu_i))^2}}$ and $y_i:= b_i \sqrt{1-(\nu_i + t(\mu_i-\nu_i))^2}$ for $i=1,\ldots,m$.
By the Cauchy-Schwarz inequality, we have
\begin{align*}
\bigg( \sum_{i=1}^m a_i b_i (\mu_i-\nu_i) \bigg)^2 &= \bigg( x_i y_i\bigg)^2 \\
&\le \bigg( \sum_{i=1}^m x_i^2 \bigg)\bigg( \sum_{i=1}^m y_i^2\bigg) \\
&= \bigg( \sum_{i=1}^m \frac{a_i^2 (\mu_i-\nu_i)^2}{1- (\nu_i +t(\mu_i-\nu_i))^2} \bigg) 
 \bigg( \sum_{i=1}^m b_i^2 \big(1- (\nu_i + t(\mu_i-\nu_i))^2\big) \bigg),
\end{align*}
which implies 
\begin{align}
\frac{\big(\sum_{i=1}^m a_i b_i (\mu_i - \nu_i)\big)^2}{\sum_{i=1}^m b_i^2\big(1 - (\nu_i + t(\mu_i-\nu_i))^2 \big)} 
 \le \sum_{i=1}^m \frac{a_i^2 (\mu_i - \nu_i)^2}{1 - (\nu_i + t(\mu_i-\nu_i))^2}. \label{eq:proof-inner-product-convexity-1}
\end{align}
Again by the Cauchy-Schwarz inequality, we have
\begin{align}
\sum_{i=1}^m b_i^2 \big( 1- (\nu_i + t(\mu_i-\nu_i))^2 \big) &= 1 - \sum_{i=1}^m b_i^2 (\nu_i + t(\mu_i-\nu_i))^2 \nonumber \\
&\le 1 - \bigg( \sum_{i=1}^m a_i b_i (\nu_i + t(\mu_i-\nu_i)) \bigg)^2. \label{eq:proof-inner-product-convexity-2}
\end{align}
Substituting \eqref{eq:proof-inner-product-convexity-2} into \eqref{eq:proof-inner-product-convexity-1} shows 
$\psi^{\prime\prime}(t) \le \phi^{\prime\prime}(t)$. By symmetry, we can also show $\psi^{\prime\prime}(t) \le \varphi^{\prime\prime}(t)$.
This completes the proof of \eqref{eq:proof-inner-product-convexity-0}.
\end{proof}

Now, we are ready to prove Theorem \ref{theorem:unbiased}. Since $f$ and $g$ are unbiased, by \eqref{eq:joint-distribution-fg-unbiased}, we have
\begin{align*}
D(P_{f(X^n)g(Y^n),\rho_0} \| P_{f(X^n)g(Y^n),\rho_1})
&= d\bigg( \frac{1+ \sum_{S\subset [n]:\atop S \neq \emptyset}  \widehat{f}(S) \widehat{g}(S) \rho_0^{|S|}}{2} \bigg\| 
  \frac{1 + \sum_{S \subset [n]: \atop S \neq \emptyset}  \widehat{f}(S) \widehat{g}(S) \rho_1^{|S|}}{2} \bigg).
\end{align*}
By noting that
\begin{align*}
\sum_{S \subset [n]: \atop S \neq \emptyset} \widehat{f}(S)^2 = \sum_{S \subset [n]: \atop S \neq \emptyset} \widehat{g}(S)^2 = 1
\end{align*}
for unbiased functions $f$ and $g$, and 
by applying Lemma \ref{lemma:inner-product-convexity} for $a = (\widehat{f}(S): S \in 2^{[n]}\backslash\{\emptyset\})$,
$b= (\widehat{g}(S): S \in 2^{[n]}\backslash\{\emptyset\})$, $\mu = (\rho_0^{|S|}: S \in 2^{[n]}\backslash\{\emptyset\})$,
and $\nu=(\rho_1^{|S|}: S \in 2^{[n]}\backslash\{\emptyset\})$, we have
\begin{align}
\lefteqn{ d\bigg( \frac{1+ \sum_{S\subset [n]:\atop S \neq \emptyset}  \widehat{f}(S) \widehat{g}(S) \rho_0^{|S|}}{2} \bigg\| 
  \frac{1 + \sum_{S \subset [n]: \atop S \neq \emptyset}  \widehat{f}(S) \widehat{g}(S) \rho_1^{|S|}}{2} \bigg) } \\
  &\le \min\bigg[
  \sum_{S \subset [n]: \atop S \neq \emptyset} \widehat{f}(S)^2 d\bigg( \frac{1+\rho_0^{|S|}}{2} \bigg\| \frac{1+\rho_1^{|S|}}{2} \bigg),
  \sum_{S \subset [n]: \atop S \neq \emptyset} \widehat{g}(S)^2 d\bigg( \frac{1+\rho_0^{|S|}}{2} \bigg\| \frac{1+\rho_1^{|S|}}{2} \bigg)
  \bigg] \label{eq:min-joint-convexity-bound} \\
  &\le \max_{1\le k \le n} d\bigg( \frac{1+\rho_0^k}{2} \bigg\| \frac{1+\rho_1^k}{2} \bigg),
\end{align}
which completes the proof. \qed

\begin{remark}
We note that Lemma \ref{lemma:inner-product-convexity} is not strictly necessary to
prove Theorem \ref{lemma:inner-product-convexity}. In fact, Theorem \ref{theorem:unbiased}
can be proved just by the joint convexity of the divergence as follows.
By the Cauchu-Schwarz inequality, we have
\begin{align}
\sum_{S \subset [n]: \atop S \neq \emptyset} |\widehat{f}(S) \widehat{g}(S)| 
\le \sqrt{\sum_{S\subset [n]} \widehat{f}(S)^2} \sqrt{\sum_{S\subset [n]} \widehat{g}(S)^2} =1.
\end{align}
Let 
\begin{align}
\alpha(S) := |\widehat{f}(S) \widehat{g}(S)|,~\beta(S) := \rom{sgn}(\widehat{f}(S) \widehat{g}(S)),
\end{align}
and $\alpha(\emptyset) := 1 - \sum_{S\subset [n]: S\neq \emptyset} \alpha(S)$.
Then, by the joint convexity of the divergence, we have
\begin{align}
D(P_{f(X^n)g(Y^n),\rho_0} \| P_{f(X^n)g(Y^n),\rho_1}) 
&= d\bigg( \frac{1+ \sum_{S\subset [n]:\atop S \neq \emptyset}  \alpha(S) \beta(S) \rho_0^{|S|}}{2} \bigg\| 
  \frac{1 + \sum_{S \subset [n]: \atop S \neq \emptyset}  \alpha(S)\beta(S) \rho_1^{|S|}}{2} \bigg) \\
 &\le \alpha(\emptyset)  d\bigg(\frac{1}{2} \bigg\| \frac{1}{2} \bigg)
  + \sum_{S\subset [n]:\atop S \neq \emptyset}  \alpha(S) d\bigg( \frac{1+\beta(S) \rho_0^{|S|}}{2} \bigg\| \frac{1+\beta(S) \rho_1^{|S|}}{2} \bigg)
  \label{eq:joint-convexity-bound} \\
  &\le \max_{1\le k \le n} d\bigg( \frac{1+\rho_0^k}{2} \bigg\| \frac{1+\rho_1^k}{2} \bigg),
\end{align}
which completes the proof of Theorem \ref{theorem:unbiased}.
Although both the bounds \eqref{eq:min-joint-convexity-bound} and \eqref{eq:joint-convexity-bound}
provide the same conclusion, they are not comparable.
\end{remark}

\subsection{Proof of Theorem \ref{theorem:biased}}

In order to prove Theorem \ref{theorem:biased}, we first investigate 
a certain property of joint distribution on $\{\pm1\} \times \{\pm 1\}$.
For $\eta \in [0,1/2]$, let $U$ be a random variable on $\{ \pm 1\}$ such that $P_U(1)=1-\eta$
and $P_U(-1) = \eta$. For $\tau \in [0,1]$, let $V$ be a random variable on $\{ \pm 1\}$ such that
$V = U$ with probability $\tau$ and $V$ is sampled independently of $U$ according to distribution $(1-\eta,\eta)$ with probability $1-\tau$.
Then, the joint distribution of $(U,V)$ is given by
\begin{align}
P_{UV,\eta,\tau} &=\left[
\begin{array}{cc}
\tau(1-\eta)+(1-\tau)(1-\eta)^2 & (1-\tau)\eta(1-\eta) \\
(1-\tau)\eta(1-\eta) & \tau \eta + (1-\tau)\eta^2 
\end{array}
\right] \\
&= \left[
\begin{array}{cc}
(1-\eta)^2 + \eta(1-\eta)\tau & \eta(1-\eta) - \eta(1-\eta)\tau \\
\eta(1-\eta) - \eta(1-\eta)\tau & \eta^2 + \eta(1-\eta)\tau
\end{array}
\right]. \label{eq:eta-tau-pair}
\end{align}
We shall call $(U,V)$ as an $\eta$-biased $\tau$-correlated pair, or just $(\eta,\tau)$-pair.

We first show the following lemma to reduce the biased case to the unbiased case.
\begin{lemma} \label{lemma:reduction-to-unbiased}
Let $\tau_0,\tau_1 \in [0,1)$. Then, for any $\eta \in [0,1/2]$, we have
\begin{align} \label{eq:maximum-unbiased}
D(P_{UV,\eta,\tau_0} \| P_{UV,\eta,\tau_1}) \le D(P_{UV,1/2,\tau_0} \| P_{UV,1/2,\tau_1}). 
\end{align}
\end{lemma}
\begin{proof}
Since $D(P_{UV,\eta,\tau_0} \| P_{UV,\eta,\tau_1})=0$ for $\eta=0$ and \eqref{eq:maximum-unbiased} trivially holds, we assume $\eta > 0$.
Denote $\bar{\tau}_c=1-\tau_c$ for $c=0,1$.
Since the marginal distributions of $P_{UV,\eta,\tau_0}$ and $P_{UV,\eta,\tau_1}$ are the same
and $P_{V|U,\eta,\tau_c}(-1|1)=\bar{\tau}_c\eta$ and $P_{V|U,\eta,\tau_c}(1|-1)=\bar{\tau}_c(1-\eta)$, by applying 
the chain rule, we have
\begin{align*}
D(P_{UV,\eta,\tau_0} \| P_{UV,\eta,\tau_1}) = (1-\eta) d(\bar{\tau}_0\eta \| \bar{\tau}_1\eta) + \eta d(\bar{\tau}_0(1-\eta) \| \bar{\tau}_1(1-\eta)) =: \Phi(\eta,\bar{\tau}_0,\bar{\tau}_1).
\end{align*}
By setting
\begin{align*}
G(\eta,\bar{\tau}_0,\bar{\tau}_1) := \frac{\partial}{\partial \eta} d(\bar{\tau}_0\eta\|\bar{\tau}_1\eta) = \bar{\tau}_0 \ln \frac{\bar{\tau}_0(1-\bar{\tau}_1\eta)}{\bar{\tau}_1(1-\bar{\tau}_0\eta)} + \frac{\bar{\tau}_1-\bar{\tau}_0}{1-\bar{\tau}_1\eta},
\end{align*}
we have
\begin{align*}
F(\eta,\bar{\tau}_0,\bar{\tau}_1) &:= \frac{\partial}{\partial \eta} \Phi(\eta,\bar{\tau}_0,\bar{\tau}_1) \\
&= - d(\bar{\tau}_0\eta\| \bar{\tau}_1\eta) + (1-\eta) G(\eta,\bar{\tau}_0,\bar{\tau}_1) + d(\bar{\tau}_0(1-\eta)\|\bar{\tau}_1(1-\eta)) - \eta G(1-\eta,\bar{\tau}_0,\bar{\tau}_1).
\end{align*}
Note that $F(\eta,\bar{\tau}_1,\bar{\tau}_1)=0$. 
In order to show that $F(\eta,\bar{\tau}_0,\bar{\tau}_1) \ge 0$, we regard $\bar{\tau}_0$ as a variable and take derivative to obtain 
\begin{align*}
F_1(\eta,\bar{\tau}_0,\bar{\tau}_1) &:= \frac{\partial}{\partial \bar{\tau}_0} F(\eta,\bar{\tau}_0,\bar{\tau}_1) \\
&= (1-2\eta)\ln \frac{\bar{\tau}_0^2(1-\bar{\tau}_1\eta)(1-\bar{\tau}_1(1-\eta))}{\bar{\tau}_1^2(1-\bar{\tau}_0\eta)(1-\bar{\tau}_0(1-\eta))} \\
&~~~ + \eta(1-\eta)\bigg[-\frac{\bar{\tau}_1}{1-\bar{\tau}_1\eta} + \frac{\bar{\tau}_1}{1-\bar{\tau}_1(1-\eta)} + \frac{\bar{\tau}_0}{1-\bar{\tau}_0\eta} - \frac{\bar{\tau}_0}{1-\bar{\tau}_0(1-\eta)} \bigg]
\end{align*}
and
\begin{align*}
F_2(\eta,\bar{\tau}_0,\bar{\tau}_1) &:= \frac{\partial^2}{\partial \bar{\tau}_0^2} F(\eta,\bar{\tau}_0,\bar{\tau}_1)  \\
&= \frac{(1-2\eta)(1-\bar{\tau}_0)(2-\bar{\tau}_0)}{\bar{\tau}_0(1-\bar{\tau}_0\eta)^2(1-\bar{\tau}_0(1-\eta))^2}.
\end{align*}
Since $F_2(\eta,\bar{\tau}_0,\bar{\tau}_1) \ge 0$ for $0<\eta\le \frac{1}{2}$ and $\bar{\tau}_0,\bar{\tau}_1\in (0,1]$, 
the function $F(\eta,\bar{\tau}_0,\bar{\tau}_1)$ is convex with respect to $\bar{\tau}_0$.
Furthermore, since $F_1(\eta,\bar{\tau}_1,\bar{\tau}_1)=0$, $F(\eta,\bar{\tau}_0,\bar{\tau}_1)$ takes global minimum $0$ at $\bar{\tau}_0=\bar{\tau}_1$.
Thus, we have $F(\eta,\bar{\tau}_0,\bar{\tau}_1) \ge 0$. Consequently, $\Phi(\eta,\bar{\tau}_0,\bar{\tau}_1)$ is non-decreasing
for $0<\eta \le \frac{1}{2}$, and we have \eqref{eq:maximum-unbiased}.
\end{proof}

Now, we are ready to prove Theorem \ref{theorem:biased}. 
Since the divergence is $0$ if $f$ is a constant function, we assume $f$ is not a constant function, which
implies $|\widehat{f}(\emptyset)| \neq 1$.
From \eqref{eq:marginal-1}, we have
\begin{align*}
\eta := \Pr\big( f(X^n)=-1) = \frac{1- \widehat{f}(\emptyset)}{2}
\end{align*}
and 
\begin{align*}
\eta(1-\eta) = \frac{1-\widehat{f}(\emptyset)^2}{4}.
\end{align*}
Without loss of generality, we can assume that $\eta \le \frac{1}{2}$; otherwise, we replace $f$ by $-f$.
By comparing \eqref{eq:joint-distribution-fg} and \eqref{eq:eta-tau-pair}, we find that, for $c= \{0,1\}$, 
the pair $(f(X^n),f(Y^n))$ induced from a $\rho_c$-correlated pair $(X^n,Y^n)$ is a$(\eta,\tau_c)$-pair with
\begin{align*}
\tau_c = \sum_{S\subset [n]:\atop S \neq \emptyset} \frac{  \widehat{f}(S)^2}{1- \widehat{f}(\emptyset)^2} \rho_c^{|S|}  \in [0,1).
\end{align*}
Thus, by applying Lemma \ref{lemma:reduction-to-unbiased}, we have
\begin{align*}
D( P_{f(X^n)f(Y^n),\rho_0} \| P_{f(X^n)f(Y^n),\rho_1})
&= D(P_{UV,\eta,\tau_0} \| P_{UV,\eta,\tau_1}) \\
&\le D(P_{UV,1/2,\tau_0} \| P_{UV,1/2,\tau_1}) \\
&= d\bigg(\frac{1+\tau_0}{2} \bigg\| \frac{1+\tau_1}{2}\bigg) \\
&\le \sum_{S \subset [n]:\atop S \neq \emptyset} \frac{ \widehat{f}(S)^2}{1- \widehat{f}(\emptyset)^2} 
d\bigg( \frac{1+\rho_0^{|S|}}{2} \bigg\| \frac{1+\rho_1^{|S|}}{2} \bigg) \\
&\le \max_{1 \le k \le n} d\bigg( \frac{1+\rho_0^k}{2} \bigg\| \frac{1+\rho_1^k}{2}\bigg),
\end{align*}
where we used the joint convexity of the KL-divergence in the second inequality. 
This completes the proof of \eqref{eq:ff}. 

Since
\begin{align*}
\Pr\big( f(X^n) = u,~-f(Y^n)=v \big) = \Pr\big( f(X^n) = u,~f(Y^n)=-v \big)
\end{align*}
and the divergence is unchanged by relabeling, \eqref{eq:negative-f} follows from \eqref{eq:ff}.
\qed

\begin{remark}
If either $\tau_0$ or $\tau_1$ is negative, \eqref{eq:maximum-unbiased} does not hold in general.
\end{remark}
\begin{remark}
Even if two functions $f$ and $g$ are different, the above proof goes through as long as 
$\widehat{f}(\emptyset)=\widehat{g}(\emptyset)$ and $\sum_{S\subset[n]:\atop S\neq\emptyset} \widehat{f}(S)\widehat{g}(S)\rho_c^{|S|}\ge0$.
\end{remark}

\subsection{Proof of Theorem \ref{theorem:local-optimality}}

Let $U=f(X^n)$, $Z=f(Y^n)$, and $V = g(Y^n)$. From \eqref{eq:affect-of-noise-operator}, we have
\begin{align*}
T_\rho f(\bm{y}) = \mathbb{E}[ f(X^n) | Y^n=\bm{y}] = \rho^k f(\bm{y}).
\end{align*}
We also have
\begin{align*}
\mathbb{E}[ f(X^n) | Y^n = \bm{y}] &= \Pr( U=1 | Y^n = \bm{y}) - \Pr( U=-1 | Y^n=\bm{y}) \\
&= 2 \Pr( U=1 | Y^n = \bm{y}) - 1 \\
&= 1 - 2 \Pr( U=-1 | Y^n=\bm{y}).
\end{align*}
Thus, we have
\begin{align*}
\Pr( U = u | Y^n = \bm{y}) = \frac{1 + u \rho^k f(\bm{y})}{2}.
\end{align*}
Note that the right hand side only depends on $f(\bm{y})$. Thus, 
$U$, $Z$, and $Y^n$ form Markov chain. Furthermore, $V$ is a function of $Y^n$,
which implies that $U$, $Z$, and $V$ form Markov chain as well.

Let 
\begin{align*}
{\cal Y}_{f,z} := \{ \bm{y} : f(\bm{y}) = z\}
\end{align*} 
and 
\begin{align*}
{\cal Y}_{g,v} := \{ \bm{y} : g(\bm{y}) = v\}.
\end{align*}
Then, 
\begin{align*}
\Pr( Y^n = \bm{y} | Z = z) = \frac{1}{|{\cal Y}_{f,z}|}
\end{align*}
and 
\begin{align*}
\Pr( V = v | Z= z) = \frac{|{\cal Y}_{f,z} \cap {\cal Y}_{g,v}|}{|{\cal Y}_{f,z}|}.
\end{align*}
Note that this probability does not depend on $\rho$, which we denote $K(v|z)$.

Now, by the data processing inequality, we have
\begin{align*}
D(P_{f(X^n)g(Y^n),\rho_0} \| P_{f(X^n)g(Y^n),\rho_1}) 
&= D(P_{UV,\rho_0} \| P_{UV, \rho_1}) \\
&= D(P_{UZ,\rho_0} \circ K \| P_{UZ,\rho_1} \circ K) \\
&\le D(P_{UZ,\rho_0} \| P_{UZ,\rho_1}) \\
&= D(P_{f(X^n)f(Y^n),\rho_0} \| P_{f(X^n)f(Y^n),\rho_1}). 
\end{align*}
\qed

\subsection{Remarks on the Remaining Case}

In this section, we briefly discuss a possible direction to tackle the remaining case, i.e., the case where $f\neq g$ are biased. 
Following the argument in \cite{PicPiaMat:18,Pichler:thesis}, 
we can consider the following relaxed problem. 
For given functions $f$ and $g$, recall the notations from Section \ref{section:Fourier}.
Without loss of generality, assume that the marginal probabilities $a_1$ and $b_1$ satisfy $\frac{1}{2} \le a_1 \le b_1 \le 1$.
Let 
\begin{align*}
{\cal P} &:= \{ S \in 2^{[n]}\backslash \{\emptyset\}: \widehat{f}(S)\widehat{g}(S) \ge 0 \}, \\
{\cal N} &:= \{ S \in 2^{[n]}\backslash \{\emptyset\}: \widehat{f}(S)\widehat{g}(S) < 0 \}
\end{align*}
be the set of subsets such that products of Fourier coefficients of $f$ and $g$ are nonnegative and negative, respectively. 
Note that
\begin{align*}
\theta_\rho &= \frac{1}{4} \sum_{S \in {\cal P}} \widehat{f}(S)\widehat{g}(S) \rho^{|S|} + \frac{1}{4} \sum_{S \in {\cal N}} \widehat{f}(S)\widehat{g}(S) \rho^{|S|}. 
\end{align*}
Since the probabilities given by \eqref{eq:joint-distribution-fg} must be nonnegative, the Fourier coefficients must satisfy
\begin{align} \label{eq:Fourier-coefficient-constraint-1}
- a_{-1}\cdot b_{-1} \le \frac{1}{4} \sum_{S \in {\cal P}} \widehat{f}(S)\widehat{g}(S) \rho^{|S|} 
+ \frac{1}{4} \sum_{S \in {\cal N}} \widehat{f}(S)\widehat{g}(S) \rho^{|S|} \le a_1 \cdot b_{-1}
\end{align}
for every $\rho \in [-1,1]$.
Furthermore, by the Cauchy-Schwarz inequality, we have
\begin{align}
\frac{1}{4} \sum_{S \in {\cal P}} \widehat{f}(S)\widehat{g}(S) - \frac{1}{4} \sum_{S \in {\cal N}} \widehat{f}(S)\widehat{g}(S)
&\le \sqrt{\sum_{ S \in 2^[n]\backslash\{\emptyset\}}\frac{ \widehat{f}(S)^2}{4}\sum_{ S \in 2^[n]\backslash\{\emptyset\}}\frac{ \widehat{g}(S)^2}{4}} \\
&= \sqrt{\frac{1-\widehat{f}(\emptyset)^2}{4} \frac{1-\widehat{g}(\emptyset)^2}{4}} \\
&= \sqrt{a_1 \cdot a_{-1} \cdot b_1 \cdot b_{-1}}. \label{eq:Fourier-coefficient-constraint-2}
\end{align}

Since it is difficult to handle the Fourier coefficients of Boolean functions, let us retain the 
conditions given by \eqref{eq:Fourier-coefficient-constraint-1} and \eqref{eq:Fourier-coefficient-constraint-2}, and consider the following set of vectors. 
For fixed disjoint sets ${\cal P}, {\cal N} \subset 2^{[n]}\backslash \{\emptyset\}$, marginal
distributions $\bm{a}=(a_1,a_{-1})$ and $\bm{b}=(b_1,b_{-1})$ satisfying 
$\frac{1}{2} \le a_1 \le b_1 \le 1$, 
let ${\cal Z}({\cal P},{\cal N},\bm{a},\bm{b})$ be the set of all
vectors $\bm{z} = (z_S: S \in {\cal P}\cup {\cal N})$ satisfying $z_S \ge 0$ for $S \in {\cal P}$, $z_S \le 0$ for $S \in {\cal N}$,
\begin{align} \label{eq:condition-z}
- a_{-1}\cdot b_{-1} \le \frac{1}{4} \sum_{S \in {\cal P}} z_S \rho^{|S|}+ \frac{1}{4} \sum_{S \in {\cal N}} z_S \rho^{|S|} \le a_1 \cdot b_{-1}
\end{align}
for every $\rho \in [-1,1]$
and
\begin{align} \label{eq:condition-z-2}
\frac{1}{4} \sum_{S \in {\cal P}} z_S - \frac{1}{4} \sum_{S \in {\cal N}} z_S 
 \le \sqrt{a_1 \cdot a_{-1} \cdot b_1 \cdot b_{-1}}.
\end{align}
Let 
\begin{align*}
P_{UV, \bm{a},\bm{b},\bm{z},\rho}(u,v) := a_u\cdot b_v + u\cdot v \cdot \theta_{\bm{z},\rho},
\end{align*}
where 
\begin{align*}
\theta_{\bm{z},\rho} := \frac{1}{4} \sum_{S \in {\cal P}} z_S \rho^{|S|} + \frac{1}{4} \sum_{S\in {\cal N}} z_S \rho^{|S|}.
\end{align*}
Under these notations, we can consider the following relaxed problem:
\begin{align} \label{eq:relaxed-optimization}
\sup_{\bm{z} \in {\cal Z}({\cal P},{\cal N},\bm{a},\bm{b})} 
 D(P_{UV, \bm{a},\bm{b},\bm{z},\rho_0} \| P_{UV, \bm{a},\bm{b},\bm{z},\rho_1})
\end{align}
for arbitrarily fixed $\bm{a}$, $\bm{b}$, ${\cal P}$, and ${\cal N}$. The optimal value of this related problem forms an upper bound on the KL-divergence between
$P_{f(X^n)g(Y^n),\rho_0}$ and $P_{f(X^n)g(Y^n),\rho_1}$ for  any Boolean functions $f$ and $g$.

When $\rho_1=0$, the argument in \cite{PicPiaMat:18,Pichler:thesis} implies that the optimal value of the 
relaxed problem \eqref{eq:relaxed-optimization} does not exceed $\ln 2 - h\big(\frac{1+\rho_0}{2}\big)$,
where $h(t)=-t\ln t-(1-t)\ln(1-t)$ is the binary entropy function. 
In fact, since $\theta_{\bm{z},0}=0$, the distribution $P_{UV, \bm{a},\bm{b},\bm{z},0}$ is a product distribution with marginals $\bm{a}$ and $\bm{b}$.
Moreover, 
the relaxed problem in \cite{PicPiaMat:18,Pichler:thesis} is derived from
the constraints \eqref{eq:condition-z} for $\rho=\rho_0,1$ and \eqref{eq:condition-z-2}. 

If we restrict our attention to $g=f$ as in Theorem \ref{theorem:biased},
then we can introduce additional constraints ${\cal N}=\emptyset$ and $\bm{a}=\bm{b}$.
For $\rho_0,\rho_1\in[0,1)$, with those additional constraints, 
we can show that the upper bound given by the relaxed problem \eqref{eq:relaxed-optimization} is tight. 
In fact, if we set $\tau_c = \frac{\theta_{\bm{z},\rho_c}}{a_1\cdot a_{-1}}$, then
we can use Lemma \ref{lemma:reduction-to-unbiased} and the joint convexity of the divergence
as in the proof of Theorem \ref{theorem:biased}. 

Unfortunately, the relaxed problem \eqref{eq:relaxed-optimization} itself
is not useful for our purpose in general $(\rho_0,\rho_1)$ and $f\neq g$. 
In fact, a vector $\bm{z} \in {\cal Z}({\cal P},{\cal N},\bm{a},\bm{b})$
may induce distributions such that $P_{UV, \bm{a},\bm{b},\bm{z},\rho_0}(1,-1)>0$ while
$P_{UV, \bm{a},\bm{b},\bm{z},\rho_1}(1,-1)=0$, and thus the value of the relaxed problem can be unbounded (a case that cannot occur in the setting of \cite{PicPiaMat:18,Pichler:thesis}). 
To derive tight bounds for general $(\rho_0,\rho_1)$ and $f\neq g$, it seems necessary to introduce 
additional conditions that Boolean functions must satisfy and that are still amenable to analysis.

\section{Fisher Information Setting} \label{section:Fisher}

In this section, we consider the problem of identifying Boolean functions that maximize the Fisher information.
For $\rho \in (-1,1)$, let $P_{X^nY^n,\rho}$ be the joint distribution of a $\rho$-correlated pair
$(X^n,Y^n)$ given by \eqref{eq:rho-correlated}. Then, for two Boolean functions
$f,g:\mathbb{F}_2^n \to \{+1,-1\}$, let $P_{f(X^n)g(Y^n),\rho}$ be the joint distribution of 
$(f(X^n),g(Y^n))$. For the parametrized family $\{ P_{f(X^n)g(Y^n),\rho} \}_{\rho \in (-1,1)}$ of distributions, 
let us consider the Fisher information given by
\begin{align*}
G(n,f,g,\rho) := \mathbb{E}\bigg[ \bigg( \frac{d}{d\rho} \ln P_{f(X^n)g(Y^n),\rho}(f(X^n),g(Y^n)) \bigg)^2 \bigg].
\end{align*}
Here we are interested in which Boolean functions $f,g$ maximize $G(n,f,g,\rho)$.
This problem was posed by Amari and Kobayashi \cite{Amari:23,Amari-biograph,KobAma:23}, and they conjectured that the Fisher information
is maximized by parity functions, i.e.,
\begin{align}
G(n,f,g,\rho) &\le \max_{S \subset [n]: \atop S \neq \emptyset} G(n, \chi_S,\chi_S,\rho) \\
&= \max_{1\le k \le n} \frac{k^2 \rho^{2(k-1)}}{1-\rho^{2k}} \label{eq:AK-conjecture}
\end{align}
for any Boolean functions $f,g$ and $\rho \in (-1,1)$.
In fact, since any level-$k$ functions have the same output distribution \eqref{eq:joint-of-level-k-function} as $\chi_S$ for $|S|=k$,
the value \eqref{eq:AK-conjecture} can be attained by any level-$k$ functions.

The Fisher information maximization problem and the divergence maximization problem can be
related by the following well known fact that the Fisher information is the second order derivative of the divergence.
\begin{lemma}
It holds that
\begin{align}
\frac{d^2}{d\rho_0^2} D(P_{f(X^n)g(Y^n),\rho_0} \| P_{f(X^n)g(Y^n),\rho})\bigg|_{\rho_0=\rho} = G(n,f,g,\rho).
\end{align}
\end{lemma}

Since 
\begin{align*}
\frac{d}{d\rho_0} D(P_{f(X^n)g(Y^n),\rho_0} \| P_{f(X^n)g(Y^n),\rho})\bigg|_{\rho_0=\rho} = 0,
\end{align*}
the Taylor expansion implies
\begin{align*}
G(n,f,g,\rho) = \lim_{\delta\to 0} \frac{2 D(P_{f(X^n)g(Y^n),\rho+\delta} \| P_{f(X^n)g(Y^n),\rho})}{\delta^2}.
\end{align*}
Thus, if we were able to prove
\begin{align} \label{eq:divergence-conjecture}
D(P_{f(X^n)g(Y^n),\rho_0} \| P_{f(X^n)g(Y^n),\rho_1}) \le 
\max_{1\le k \le n} d\bigg( \frac{1+\rho_0^k}{2} \bigg\| \frac{1+\rho_1^k}{2} \bigg)
\end{align}
for any Boolean functions $f,g$ and $\rho_0,\rho_1\in (-1,1)$,\footnote{In fact, for the purpose of proving \eqref{eq:AK-conjecture},
it suffices to prove \eqref{eq:divergence-conjecture} when $\rho_0$ and $\rho_1$ are close.} 
then we have
\begin{align*}
G(n,f,g,\rho) &= \lim_{\delta\to 0} \frac{2 D(P_{f(X^n)g(Y^n),\rho+\delta} \| P_{f(X^n)g(Y^n),\rho})}{\delta^2} \\
&\le \lim_{\delta\to 0} \max_{1\le k \le n} \frac{2d\big( \frac{1+(\rho+\delta)^k}{2} \big\| \frac{1+\rho^k}{2} \big)}{\delta^2} \\
&= \max_{1\le k \le n} \lim_{\delta\to 0} \frac{2d\big( \frac{1+(\rho+\delta)^k}{2} \big\| \frac{1+\rho^k}{2} \big)}{\delta^2} \\
&= \max_{1\le k \le n} \frac{k^2 \rho^{2(k-1)}}{1-\rho^{2k}}.
\end{align*}

From the partial solutions to the divergence maximization problem (Theorem \ref{theorem:unbiased} and Theorem \ref{theorem:biased}),
we can derive the following partial solutions to the Amari-Kobayashi conjecture. For the sake of completeness, we also provide direct
proofs in Sections \ref{subsec:proof-Fisher-unbiased} and \ref{subsec:proof-Fisher-biased}.
\begin{theorem}[Unbiased Functions] \label{theorem:Fisher-unbiased}
Let $\rho \in (-1,1)$. Let $f,g:\mathbb{F}_2^n\to\{\pm1\}$ be unbiased Boolean functions, i.e.,
$\widehat{f}(\emptyset)=\widehat{g}(\emptyset)=0$. Then, we have
\begin{align*}
G(n,f,g,\rho) \le \max_{1\le k \le n} \frac{k^2 \rho^{2(k-1)}}{1-\rho^{2k}}.
\end{align*}
\end{theorem}

\begin{theorem}[Biased Functions] \label{theorem:Fisher-biased}
Let $\rho \in [0,1)$.\footnote{For $\rho=0$, it does not follow from Theorem \ref{theorem:biased} since 
Theorem \ref{theorem:biased} cannot be applied for negative $\rho_0$ in a neighborhood of $\rho_1=0$. 
However, this case can be handled via direct proof; see Section \ref{theorem:Fisher-biased}.} 
Let $f:\mathbb{F}_2^n \to \{\pm 1\}$ be a Boolean function (not necessarily unbiased). Then, we have
\begin{align*}
G(n,f,f,\rho) \le \max_{1\le k \le n} \frac{k^2 \rho^{2(k-1)}}{1-\rho^{2k}}.
\end{align*}
\end{theorem}

\subsection{Proof of Theorem \ref{theorem:Fisher-unbiased}} \label{subsec:proof-Fisher-unbiased}

By using the notation in \eqref{eq:joint-distribution-fg}, we have
\begin{align*}
\frac{d}{d\rho} P_{f(X^n)g(Y^n),\rho}(u,v) = u\cdot v\cdot \frac{d}{d\rho} \theta_\rho 
= \frac{uv}{4} \sum_{S\subset [n]: \atop S\neq \emptyset} |S| \widehat{f}(S) \widehat{g}(S) \rho^{|S|-1}.
\end{align*}
Then, we can write
\begin{align*}
G(n,f,g,\rho) &= \sum_{u,v \in \{\pm 1\}} \frac{\bigg( \frac{d}{d\rho} P_{f(X^n)g(Y^n),\rho}(u,v) \bigg)^2}{P_{f(X^n)g(Y^n),\rho}(u,v)} \\
&= \sum_{u,v \in \{\pm 1\}} \frac{1}{4} \frac{\big( uv \sum_{S: S\neq \emptyset} |S| \widehat{f}(S) \widehat{g}(S) \rho^{|S|-1}\big)^2}{4 a_u b_v + uv \sum_{S: S\neq \emptyset} \widehat{f}(S) \widehat{g}(S) \rho^{|S|}}.
\end{align*}

When $f$ and $g$ are unbiased, since $a_u = b_v = \frac{1}{2}$, we have
\begin{align*}
G(n,f,g,\rho) &= \sum_{u,v \in \{\pm 1\}} \frac{1}{4} \frac{\big( uv \sum_{S: S\neq \emptyset} |S| \widehat{f}(S) \widehat{g}(S) \rho^{|S|-1}\big)^2}{1 + uv \sum_{S: S\neq \emptyset} \widehat{f}(S) \widehat{g}(S) \rho^{|S|}}  \\
&= \frac{\big( \sum_{S: S\neq \emptyset} |S| \widehat{f}(S) \widehat{g}(S) \rho^{|S|-1}\big)^2}{1 - \big(\sum_{S: S\neq \emptyset} \widehat{f}(S) \widehat{g}(S) \rho^{|S|}\big)^2}.
\end{align*}

Let 
\begin{align}
K := \sum_{S: S\neq \emptyset}  |\widehat{f}(S) \widehat{g}(S)|.
\end{align}
If $K=0$, then the Fisher information is $0$. Thus, we assume $K > 0$. 
By the Cauchy-Schwarz inequality, we have
\begin{align}
K &\le \sqrt{\sum_S \widehat{f}(S)^2} \sqrt{\sum_S \widehat{g}(S)^2} =1.
\end{align}
Let 
\begin{align}
\alpha(S) := \frac{|\widehat{f}(S) \widehat{g}(S)|}{K},~\beta(S) := \rom{sgn}(\widehat{f}(S) \widehat{g}(S)).
\end{align}
Then, we can bound the Fisher information as
\begin{align}
G(n,f,g,\rho) &= \frac{K^2 \big(  \sum_{S: S\neq \emptyset} \alpha(S) \beta(S) |S|  \rho^{|S|-1}\big)^2}{1 - K^2 \big(\sum_{S: S\neq \emptyset} \alpha(S) \beta(S) \rho^{|S|}\big)^2} \\
&\le \frac{K^2 \sum_{S: S\neq \emptyset} \alpha(S)  |S|^2  \rho^{2(|S|-1)}}{1 - K^2 \sum_{S: S\neq \emptyset} \alpha(S)  \rho^{2|S|}} \label{eq:Fisher-unbiased-proof-1} \\
&\le \frac{K^2 \sum_{S: S\neq \emptyset} \alpha(S)  |S|^2  \rho^{2(|S|-1)}}{K^2 - K^2 \sum_{S: S\neq \emptyset} \alpha(S)  \rho^{2|S|}}  \label{eq:Fisher-unbiased-proof-2}  \\
&= \frac{ \sum_{S: S\neq \emptyset} \alpha(S)  |S|^2  \rho^{2(|S|-1)}}{  \sum_{S: S\neq \emptyset} \alpha(S) (1 -  \rho^{2|S|})} \\
&\le \max_{S: S\neq \emptyset} \frac{|S|^2 \rho^{2(|S|-1)}}{1 - \rho^{2|S|}} \label{eq:Fisher-unbiased-proof-3} \\
&= \max_{1\le k \le n} \frac{k^2 \rho^{2(k-1)}}{1-\rho^{2k}},
\end{align}
where the inequality \eqref{eq:Fisher-unbiased-proof-1} follows from the convexity 
of $t \mapsto t^2$ and $\beta(S)^2=1$ (used twice for the numerator and denominator, respectively),
the inequality \eqref{eq:Fisher-unbiased-proof-2} follows from $K^2 \le 1$, and the inequality \eqref{eq:Fisher-unbiased-proof-3} follows from
\begin{align} \label{eq:convex-max}
\frac{\lambda A_0 + (1-\lambda) A_1}{\lambda B_0 + (1-\lambda)B_1} \le \max\bigg[ \frac{A_0}{B_0}, \frac{A_1}{B_1} \bigg]
\end{align}
for $A_0,A_1,B_0,B_1 > 0$ and $0\le \lambda\le1$.
\qed

\subsection{Proof of Theorem \ref{theorem:Fisher-biased}} \label{subsec:proof-Fisher-biased}

If $f$ is a constant function, then the Fisher information is $0$. Thus, we assume 
that $f$ is not a constant function, which implies $|\widehat{f}(\emptyset)|\neq 1$. 
Let
\begin{align*}
\eta := \Pr(f(X^n) = +1)
\end{align*}
and
\begin{align*}
\tau(\rho) &:= \frac{1}{4\eta(1-\eta)} \sum_{S: S \neq \emptyset} \widehat{f}(S)^2 \rho^{|S|} \\
&= \frac{1}{1-\widehat{f}(\emptyset)^2}  \sum_{S: S \neq \emptyset} \widehat{f}(S)^2 \rho^{|S|}. 
\end{align*}
Note that
\begin{align*}
\tau^\prime(\rho) = \frac{1}{4\eta(1-\eta)} \sum_{S: S \neq \emptyset} \widehat{f}(S)^2 |S| \rho^{|S|-1}.
\end{align*}
We can write 
\begin{align}
G(n,f,f,\rho) &= \bigg[ \frac{(1-\eta)^2\eta^2}{(1-\eta)^2 + (1-\eta)\eta \tau(\rho)} 
 + \frac{(1-\eta)^2\eta^2}{\eta^2+(1-\eta)\eta \tau(\rho)} + \frac{2(1-\eta)^2\eta^2}{(1-\eta)\eta - (1-\eta)\eta \tau(\rho)}  \bigg] \tau^\prime(\rho)^2.
 \label{eq:Fisher-biased-proof-1}
\end{align}
Since $\tau(0)=0$, the bracket factor in \eqref{eq:Fisher-biased-proof-1} is $1$; furthermore, 
\begin{align*}
\tau^\prime(0) = \frac{1}{4\eta(1-\eta)} \sum_{S: |S| =1} \widehat{f}(S)^2 \le 1.
\end{align*}
Thus, we have $G(n,f,f,0) \le 1$ for any function $f$, which completes the
proof for $\rho=0$. 

In the rest of the proof, we assume $\rho>0$.
We first show the following to reduce the biased case to the unbiased case.
\begin{lemma} \label{lemma:Fisher-proof-biased-unbiased}
For fixed $\tau \in (0,1)$, let
\begin{align*}
F(\eta) := \frac{(1-\eta)^2\eta^2}{(1-\eta)^2 + (1-\eta)\eta \tau} 
 + \frac{(1-\eta)^2\eta^2}{\eta^2+(1-\eta)\eta \tau} + \frac{2(1-\eta)^2\eta^2}{(1-\eta)\eta - (1-\eta)\eta \tau}. 
\end{align*}
Then, we have
\begin{align*}
F(\eta) \le F(1/2) = \frac{1}{1-\tau^2}.
\end{align*}
\end{lemma}
\begin{proof}
Note that $F(0)=0$. We can verify
\begin{align*}
F^\prime(\eta) = \frac{(1-2\eta) \tau(1+\tau)}{(1-\eta(1-\tau))^2(1-\tau)(\eta+\tau-\eta\tau)^2}.
\end{align*}
Thus, $F^\prime(\eta) \ge 0$ for $\eta\in [0,1/2]$ and $F^\prime(\eta)\le 0$ for $\eta \in [1/2,1]$,
and thus $F(\eta)$ takes the maximum at $\eta=1/2$.
\end{proof}

By applying Lemma \ref{lemma:Fisher-proof-biased-unbiased} to \eqref{eq:Fisher-biased-proof-1}, we can bound the Fisher information as
\begin{align}
G(n,f,f,\rho) &\le \frac{\tau^\prime(\rho)^2}{1-\tau(\rho)^2} \\
&= \frac{\bigg( \frac{1}{1-\widehat{f}(\emptyset)^2} \sum_{S: S \neq \emptyset} \widehat{f}(S)^2 |S| \rho^{|S|-1} \bigg)^2}{1 - \bigg(\frac{1}{1- \widehat{f}(\emptyset)^2} \sum_{S: S \neq \emptyset} \widehat{f}(S)^2 \rho^{|S|} \bigg)^2} \\
&\le \frac{\frac{1}{1-\widehat{f}(\emptyset)^2} \sum_{S: S \neq \emptyset} \widehat{f}(S)^2 |S|^2 \rho^{2(|S|-1)}}{1 - \frac{1}{1- \widehat{f}(\emptyset)^2} \sum_{S: S \neq \emptyset} \widehat{f}(S)^2 \rho^{2 |S|} } \label{eq:Fisher-biased-proof-2} \\
&\le \max_{S : S \neq \emptyset} \frac{|S|^2 \rho^{2(|S|-1)}}{1-\rho^{2 |S|} } \label{eq:Fisher-biased-proof-3} \\
&=\max_{1\le k\le n} \frac{k^2 \rho^{2(k-1)}}{1-\rho^{2k}},
\end{align}
where the inequality \eqref{eq:Fisher-biased-proof-2} follows from the convexity of $t \mapsto t^2$, 
and the inequality \eqref{eq:Fisher-biased-proof-3} follows from \eqref{eq:convex-max}. \qed

\section{Bayesian Hypothesis Testing}
\label{section:Bayes}


In this section, we consider Bayesian hypothesis testing with uniform prior distribution.
Two encoders observe $X^n$ and $Y^n$ on $\mathbb{F}_2^n$, respectively, and send one-bit messages 
$f(X^n)$ and $g(Y^n)$ to a receiver using Boolean functions $f,g:\mathbb{F}_2^n\to \{\pm 1\}$.
Then, the receiver decides the hypothesis using a function $\phi:\{\pm 1\}^2 \to \{\mathtt{H}_0,\mathtt{H}_1\}$;
under the hypothesis $\mathtt{H}_c$ for $c\in \{0,1\}$, the observation $(X^n,Y^n)$
is a $\rho_c$-correlated pair (cf.~\eqref{eq:rho-correlated}). Thus, a testing scheme is described by a triplet $(f,g,\phi)$,
and the correct probability is given by
\begin{align} \label{eq:correct-probability-Bayesian}
\Pc(f,g,\phi) := \sum_{c\in \{0,1\}} \frac{1}{2} \Pr_{\rho_c}\big( \phi(f(X^n),g(Y^n)) = \mathtt{H}_c \big).
\end{align}
We are interested in maximizing the correct probability. We first  determine the optimal decision rule for any given encoding functions $f,g$. 

\begin{proposition}\label{prop:decisionrule}
Given any Boolean encoding functions $f,g$, the optimal decision rule $\phi^*$ that maximizes the correct probability is given by 
\begin{align}
\phi^*(u,v) = \left\{
\begin{array}{ll}
\mathtt{H}_0 & \mbox{if } u\cdot v = 1 \\
\mathtt{H}_1 & \mbox{if } u \cdot v = -1
\end{array}
\right. \label{eq:Bayesian-decision-OPT1}
\end{align}
when $\mathbb{E}_{\rho_{0}}[f(X^{n})g(Y^{n})]\ge \mathbb{E}_{\rho_{1}}[f(X^{n})g(Y^{n})]$, and
\begin{align}
\phi^*(u,v) = \left\{
\begin{array}{ll}
\mathtt{H}_1 & \mbox{if } u\cdot v = 1 \\
\mathtt{H}_0 & \mbox{if } u \cdot v = -1
\end{array}
\right. \label{eq:Bayesian-decision-OPT2}
\end{align}
when $\mathbb{E}_{\rho_{0}}[f(X^{n})g(Y^{n})] < \mathbb{E}_{\rho_{1}}[f(X^{n})g(Y^{n})]$. 
Moreover, the correct probability induced by $(f,g,\phi^*)$ is 
\begin{align}  
\Pc(f,g,\phi^*) & = \frac{1}{2} + |p_0 - p_1 | \\
& = \frac{1}{2} + \frac{1}{4} \big|\mathbb{E}_{\rho_{0}}[f(X^{n})g(Y^{n})]-\mathbb{E}_{\rho_{1}}[f(X^{n})g(Y^{n})]\big|,\label{eq:Pc-OPT}
\end{align}
where $p_c= \Pr_{\rho_c}\big( f(X^n)= g(Y^n)= 1 \big)$ for $c \in \{0,1\}.$
\end{proposition}
\begin{proof}
We compute the correct probability for all possible decision rules, and then compare them to determine the optimal one. 
By symmetry, without loss of generality, we assume 
\begin{align}
\Delta:=\mathbb{E}_{\rho_{0}}[f(X^{n})g(Y^{n})] - \mathbb{E}_{\rho_{1}}[f(X^{n})g(Y^{n})] \ge 0. 
\end{align}
Note that $\Delta= 4 (p_0 -p_1)$.

We divide all possible decision rules into the following four cases. 

Case 1 (Diagonal Partition): For the decision rule   which outputs  $\mathtt{H}_0$ if and only if $u\cdot v = 1$ (i.e., the one  
given in  \eqref{eq:Bayesian-decision-OPT1}),
the  induced correct probability is 
$\Pc = \frac{1}{2} + p_0 - p_1 = \frac{1}{2} + \frac{\Delta}{4}$, 
which is exactly that given in \eqref{eq:Pc-OPT}. 
Similarly, for the decision rule  which outputs  $\mathtt{H}_0$ if and only if $u\cdot v = -1$, 
the  induced correct probability is 
$\Pc = \frac{1}{2} - \frac{\Delta}{4}$. 

Case 2 (Horizontal or Vertical Partition): For the decision rule which outputs \(\mathtt{H}_0\) if and only if \(u = 1\) (analogously, rules where \(\mathtt{H}_0\) is selected iff \(u=-1\), \(v=1\), or \(v=-1\)),
the induced correct probability is  
$\Pc = \frac{1}{2}$. 

Case 3 (Corner Partition): For the decision rule which outputs \(\mathtt{H}_0\) if and only if \(u = v= 1\), 
the  induced correct probability is  
$\Pc = \frac{1+ p_0 - p_1}{2}  = \frac{1}{2} + \frac{\Delta}{8}$.
Similarly,   given any $s,t\in \{\pm 1\}$, for the decision rule which outputs $\mathtt{H}_0$ if and only if $s u = t v = 1$, the  induced correct probability is either  
$\Pc = \frac{1}{2} + \frac{\Delta}{8}$ or $\Pc = \frac{1}{2} - \frac{\Delta}{8}$. 

Case 4 (Trivial Partition): For the decision rule which always outputs \(\mathtt{H}_0\) regardless of the values of $u$ and $v$,
the  induced correct probability is  
$\Pc = \frac{1}{2}$.
Similarly,  for the decision rule which always outputs \(\mathtt{H}_1\) regardless of the values of $u$ and $v$,
the  induced correct probability is  also 
$\Pc = \frac{1}{2}$.

Comparing the correct probabilities computed above, one can find that the largest one is $\Pc =  \frac{1}{2} + \frac{\Delta}{4}$, which induced by the  decision rule
given in  \eqref{eq:Bayesian-decision-OPT1}. 
\end{proof}

\begin{theorem} \label{theorem:Bayesian-Boolean}
For $\rho_0, \rho_1 \in [-1,1]$ and any testing scheme $(f,g,\phi)$, we have
\begin{align} \label{eq:Bayesian-boolean-cube}
\Pc(f,g,\phi) \le \frac{1}{2} + \frac{1}{4} \max_{1\le k \le n} | \rho_0^k - \rho_1^k|.
\end{align}
Furthermore, when the maximum in the right hand side of \eqref{eq:Bayesian-boolean-cube}
is attained by $k=k^\star$, then the equality in \eqref{eq:Bayesian-boolean-cube} is attained 
by any level-$k^\star$ function $f=g$, and 
\begin{align}
\phi(u,v) = \left\{
\begin{array}{ll}
\mathtt{H}_0 & \mbox{if } u\cdot v = 1 \\
\mathtt{H}_1 & \mbox{if } u \cdot v = -1
\end{array}
\right. \label{eq:Bayesian-decision-1}
\end{align}
when $\rho_0^{k^\star} \ge \rho_1^{k^\star}$ and
\begin{align}
\phi(u,v) = \left\{
\begin{array}{ll}
\mathtt{H}_1 & \mbox{if } u\cdot v = 1 \\
\mathtt{H}_0 & \mbox{if } u \cdot v = -1
\end{array}
\right. \label{eq:Bayesian-decision-2}
\end{align}
when $\rho_0^{k^\star} < \rho_1^{k^\star}$.
\end{theorem}
\begin{proof}
By using \eqref{eq:joint-distribution-fg} and Proposition \ref{prop:decisionrule}, we can rewrite the correct probability as
\begin{align} \label{eq:proof-Bayesian-2} 
\Pc(f,g,\phi) \le \frac{1}{2} +  |\theta_{\rho_0} - \theta_{\rho_1} |,
\end{align}
where the equality holds if we choose $\phi$ as the optimal decision rule $\phi^*$ given in Proposition \ref{prop:decisionrule}.
Now, by using \eqref{eq:theta-definition-2}, we have
\begin{align} \label{eq:proof-Bayesian-3}
|\theta_{\rho_0} - \theta_{\rho_1} | \le \frac{1}{4} \sum_{S \subset [n]: \atop S \neq \emptyset} | \widehat{f}(S) \widehat{g}(S)| \cdot | \rho_0^{|S|}-\rho_1^{|S|}|.
\end{align}
By the Cauchy-Schwarz inequality and the Parseval identity \eqref{eq:Parseval-Boolean}, we have
\begin{align*}
\sum_{S \subset [n]: \atop S \neq \emptyset} | \widehat{f}(S) \widehat{g}(S)| \le \sqrt{ \sum_{S\subset[n]} \widehat{f}(S)^2 }
\sqrt{ \sum_{S \subset [n]} \widehat{g}(S)^2 } =1.
\end{align*}
Thus, the right hand side of \eqref{eq:proof-Bayesian-3} can be regarded as an average of
$| \rho_0^{|S|}-\rho_1^{|S|}|$ over $S \subset [n]$, and we have
\begin{align} \label{eq:proof-Bayesian-4}
\sum_{S \subset [n]: \atop S \neq \emptyset} | \widehat{f}(S) \widehat{g}(S)| \cdot | \rho_0^{|S|}-\rho_1^{|S|}|
\le \max_{1\le k \le n} |\rho_0^k - \rho_1^k|.
\end{align}
By combining \eqref{eq:proof-Bayesian-2}, \eqref{eq:proof-Bayesian-3}, and \eqref{eq:proof-Bayesian-4},  
we have \eqref{eq:Bayesian-boolean-cube}. The equality condition for \eqref{eq:Bayesian-boolean-cube}
follows from the equality condition for \eqref{eq:proof-Bayesian-2} and that $\theta_{\rho_c} = \frac{1}{4} \rho_c^k$
for a level-$k$ function $f=g$.
\end{proof}

\subsection{Maximal Correlation Difference} \label{subsec:MCD}

The maximal correlation was used by Witsenhausen \cite{witsenhausen:75} to derive an impossibility bound for the non interactive correlation distillation (NICD) problem.
Particularly, for binary double symmetric sources, the probability of agreement among
unbiased functions is maximized by dictator functions. 
Proposition \ref{prop:decisionrule} tells us that when the decision rule is applied in our Bayesian hypothesis testing above, the correct probability of the test 
is given by the difference of the expectations of $fg$ under the two hypotheses.  This prompts us to introduce the following concept as a generalization of maximal correlation.


For finite alphabets ${\cal X}$ and ${\cal Y}$, let $P_{X^nY^n,\rho_0}$ and $P_{X^nY^n,\rho_1}$ be
distributions on ${\cal X}^n\times {\cal Y}^n$ such that
$X^n$ marginal distributions satisfy $P_{X^n,\rho_0}=P_{X^n,\rho_1}=P_{X^n}$ and $Y^n$ marginal distribution
satisfy $P_{Y^n,\rho_0}=P_{Y^n,\rho_1}=P_{Y^n}$. Then, we define the maximal correlation difference (MCD) between the two
distributions as
\begin{align*}
\mathtt{MCD}(P_{X^nY^n,\rho_0},P_{X^nY^n,\rho_1}) := 
\sup_{f,g} \sum_{x^n,y^n}\big( P_{X^nY^n,\rho_0}(x^n,y^n) - P_{X^nY^n,\rho_1}(x^n,y^n) \big)f(x^n)g(y^n),
\end{align*}
where the supremum is taken over all functions $f:{\cal X}^n\to\mathbb{R}$ and ${\cal Y}^n\to\mathbb{R}$
satisfying $\mathbb{E}[f(X^n)]=\mathbb{E}[g(Y^n)]=0$ and $\mathbb{E}[f(X^n)^2]=\mathbb{E}[g(Y^n)^2]=1$.

We consider 
\[
\mathbf{D}:=\left[\frac{P_{X^{n}Y^{n},\rho_{0}}(x^{n},y^{n})-P_{X^{n}Y^{n},\rho_{1}}(x^{n},y^{n})}{\sqrt{P(x^{n})P(y^{n})}}\right]_{(x^{n},y^{n})}
\]
as a matrix, and 
\[
\vec{u}:=(f(x^{n})\sqrt{P(x^{n})})_{x^{n}},\qquad\vec{v}:=(g(y^{n})\sqrt{P(y^{n})})_{y^{n}}
\]
as vectors. Then, we can rewrite $\mathtt{MCD}$ as 
\[
\mathtt{MCD}(P_{X^nY^n,\rho_0},P_{X^nY^n,\rho_1})=\sup_{\vec{u},\vec{v}:\|\vec{u}\|_{2}=\|\vec{v}\|_{2}=1,\vec{u}\cdot\sqrt{P_{X^{n}}}=\vec{v}\cdot\sqrt{P_{Y^{n}}}=0}\vec{u}^{\top}\mathbf{D}\vec{v}.
\]
Noting that $\sqrt{P_{X^{n}}}$ and $\sqrt{P_{Y^{n}}}$ are the left-singular and right-singular vectors of $\mathbf{D}$ for the singular value $0$ (the smallest  singular value).
Thus, the constraints $\vec{u}\cdot\sqrt{P_{X^{n}}}=\vec{v}\cdot\sqrt{P_{Y^{n}}}=0$ are equivalent to that the $\vec{u},\vec{v}$ are respectively in the spaces  spanned by the remaining left-singular and right-singular vectors. Obviously, the supremum is attained by the left-singular and right-singular vectors for the largest singular value, which implies that 
\[
\mathtt{MCD}(P_{X^nY^n,\rho_0},P_{X^nY^n,\rho_1})=\sigma_{1}(\mathbf{D})=\sqrt{\lambda_{1}(\mathbf{D}^{\top}\mathbf{D})}.
\]
Here $\sigma_{1}(\mathbf{D})$ is the largest singular value of $\mathbf{D}$,
and $\lambda_{1}(\mathbf{D}^{\top}\mathbf{D})$ is the largest eigenvalue
of $\mathbf{D}^{\top}\mathbf{D}$. This singular value characterization is
analogous to that for maximal correlation given in  \cite{anantharam2013maximal}. 

By the definition of $\mathtt{MCD}$, we have the following observation.
\begin{proposition}\label{prop:mcd}
For any $\{\pm1\}$-valued functions $f,g$, 
\begin{align*}
|\mathbb{E}_{\rho_{0}}[f(X^{n})g(Y^{n})]-\mathbb{E}_{\rho_{1}}[f(X^{n})g(Y^{n})]| & \le \mathtt{MCD}(P_{X^nY^n,\rho_0},P_{X^nY^n,\rho_1}).
\end{align*}
\end{proposition}
\begin{proof}
  Since the left hand side of the bound is $0$ if either $f$ or $g$ is constant, assume that both functions are not constant. 
  Let $\mu_1,\mu_2$ be the means of $f,g$, respectively, and   $V_1,V_2$ be the variances of $f,g$, respectively. Then, 
  by letting $\tilde{f}=\frac{f-\mu_1}{\sqrt{V_1}}$ and $\tilde{g}=\frac{g-\mu_2}{\sqrt{V_2}}$,
  we have 
  \begin{align*}
  \mathbb{E}_{\rho_{0}}[f(X^{n})g(Y^{n})]-\mathbb{E}_{\rho_{1}}[f(X^{n})g(Y^{n})]
  = \sqrt{V_1V_2}
  \big(\mathbb{E}_{\rho_{0}}[\tilde{f}(X^{n}) \tilde{g}(Y^{n})]-\mathbb{E}_{\rho_{1}}[\tilde{f}(X^{n}) \tilde{g}(Y^{n})]\big).
  \end{align*}
  Since $f$ and $g$ are Boolean, we have $V_1,V_2 \le 1$. Thus, the claim follows from the definition of the MCD.
\end{proof}

This upper bound in  Proposition \ref{prop:mcd}  is not sharp in general, since the $\mathtt{MCD}$ is not always attained by $\{\pm1\}$-valued functions. However, for the binary case, it is indeed sharp. 
We now examine this point. 
For this case, 
\begin{align*}
P_{X^{n}Y^{n},\rho} & =2^{-n}\mathbf{T}^{\otimes n}_{\rho},
\end{align*}
where $\mathbf{T}_{\rho}$ is the semigroup matrix 
\[
\mathbf{T}_{\rho}=\begin{bmatrix}\frac{1+\rho}{2} & \frac{1-\rho}{2}\\
\frac{1-\rho}{2} & \frac{1+\rho}{2}
\end{bmatrix}=\mathbf{U}\Lambda\mathbf{U}^{\top},
\]
with $\mathbf{U}=\frac{1}{\sqrt{2}}\begin{bmatrix}1 & 1\\
1 & -1
\end{bmatrix}$ and $\Lambda=\begin{bmatrix}1 & 0\\
0 & \rho
\end{bmatrix}$. Hence, 
\begin{align*}
\mathbf{D} & =\mathbf{T}^{\otimes n}_{\rho_{0}}-\mathbf{T}^{\otimes n}_{\rho_{1}}\\
 & =\mathbf{U}^{\otimes n}(\begin{bmatrix}1 & 0\\
0 & \rho_{0}
\end{bmatrix}^{\otimes n}-\begin{bmatrix}1 & 0\\
0 & \rho_{1}
\end{bmatrix}^{\otimes n})(\mathbf{U}^{\otimes n})^{\top},
\end{align*}
whose eigenvalues  are $\rho^{k}_{0}-\rho^{k}_{1}$
with multiplicity ${n \choose k}$. This yields $\mathtt{MCD}=\sigma_{1}(\mathbf{D})=\max_{1\le k\le n}|\rho^{k}_{0}-\rho^{k}_{1}|$. 
As observed in Theorem \ref{theorem:Bayesian-Boolean}, this value is attained by any level-$k^\star$ $\{\pm1\}$-valued function $f=g$, where $k^\star$ is the optimal value attaining the maximum above. Therefore, the upper bound in Proposition \ref{prop:mcd} is sharp for the binary case. 

\section{Discussion} \label{section:discussion}

In this paper, we considered the problem of identifying the most discriminative Boolean (MDBF) functions.
For the problem of maximizing divergence, we proved that level-$k$ functions are optimal
when $f$ and $g$ are unbiased, or when $f=g$. For the problem of maximizing Fisher information, we also proved
that level-$k$ functions are optimal under the same conditions. The latter result partially resolve the
conjecture of Amari and Kobayashi \cite{Amari:23,Amari-biograph,KobAma:23}. In the Bayesian hypothesis testing framework, we proved that
level-$k$ functions are optimal among all functions.

When $\rho_1=0$, the problem of maximizing divergence can be viewed as a
two function version of the most informative Boolean function (MIBF) problem.
The two function version of the MIBF problem has already resolved \cite{PicPiaMat:18,Pichler:thesis}, with dictator functions being optimal.
However, for the MDBF problem studied in this paper, the general solution is open. 

A natural extension of the problem studied in this paper is the one function
version of the MDBF problem, i.e., the problem of identifying a Boolean function $f$ that maximizes the divergence
\begin{align*}
D(P_{f(X^n)Y^n,\rho_0}\|P_{f(X^n)Y^n,\rho_1}).
\end{align*}
This problem can be regarded as a generalization of the Courtade-Kumar conjecture \cite{CouKum:14}.
For the one function version of the MIBF problem, it has been conjectured that dictator 
functions are optimal. However, the one function version of the MDBF problem requires 
more careful consideration.

From the partial results obtained in this paper for the two function version,
one might guess that level-$k$ functions are optimal at least among unbiased functions.
However, there are counterexamples for some parameter choices. For instance, 
when $n=3$ and the parameters are $(\rho_0,\rho_1)=(0.3,0.95)$ or $(\rho_0,\rho_1)=(0.95,-0.95)$, we can numerically verify that
the majority function is optimal; but for $(\rho_0,\rho_1)=(0.4,0.95)$, we can numerically verify that level-$2$ functions are optimal. 
Thus, for parameter regimes where $\rho_0 < \rho_1$, or where $\rho_0$ and $\rho_1$
have opposite signs, we currently do not have a clear candidate for the optimal function.
On the other hand, in the regime $0 < \rho_1 < \rho_0 < 1$, we have so far found no
example in which a function other than a level-$k$ function is optimal. 
Similarly, for the one function version of the Fisher information maximization problem,
we have so far found no example that outperforms level-$k$ functions.

\section*{Acknowledgment}

GPT-5.5 Pro was used to assist with calculations in the proof of Lemma 2. Certain ideas leading to the proof of 
Theorem 3 were developed in dialogue with GPT-5.5 Pro. It was also used to assist in editing the manuscript.
SW was supported in part by the Japan Society for the Promotion of Science (JSPS) KAKENHI under Grant
26K02864 and 26H02489, and by JST, CRONOS, Japan Grant
Number JPMJCS25N5.


\bibliographystyle{./IEEEtranS}
\bibliography{Reference}

\end{document}